\documentclass[11pt, letterpaper]{article}

\usepackage{graphicx}
\usepackage{amsmath}
\usepackage{amssymb}
\usepackage{setspace}
\usepackage{epstopdf}
\usepackage{color}
\usepackage[letterpaper,left=1in,right=1in,top=1in,bottom=1in]{geometry}
\usepackage[linesnumbered,ruled,vlined,longend]{algorithm2e}
\usepackage{multirow}
\usepackage{longtable}
\usepackage{rotating}
\usepackage{threeparttable}
\usepackage{colortbl}
\usepackage{amsthm}
\usepackage[modulo]{lineno}
\usepackage{enumerate}

\newtheorem{lemma}{Lemma}
\newtheorem{proposition}{Proposition}

\newtheorem{theorem}{Theorem}

\newtheorem{definition}{Definition}
\newtheorem{assumption}{Assumption}
\newtheorem{example}{Example}

\usepackage{xpatch}
\makeatletter
\xpatchcmd{\@thm}{\thm@headpunct{.}}{\thm@headpunct{}}{}{}
\makeatother

\title{\large \bf Computing robust control invariant sets of constrained nonlinear systems: \\ A graph algorithm approach}

\author{
    \centerline{\normalsize Benjamin Decardi-Nelson$^{a}$, Jinfeng Liu$^{a,}$
    \thanks{Corresponding author: J. Liu. Tel: +1-780-492-1317. Fax: +1-780-492-2881. Email: jinfeng@ualberta.ca}}
    \vspace{5mm} \\
    \centerline{\small $^{a}$ Department of Chemical \& Materials Engineering, University of Alberta,}\\
    \centerline{\small Edmonton AB T6G 1H9, Canada}
}

\allowdisplaybreaks

\begin{document}

%\doublespacing
%\onehalfspacing
\date{}
\maketitle
\setstretch{1.39}

\begin{abstract}

	This paper deals with the computation of the largest robust control invariant sets (RCISs) of constrained nonlinear systems. The proposed approach is based on casting the search for the invariant set as a graph theoretical problem.  Specifically, a general class of discrete-time time-invariant nonlinear systems is considered. First, the dynamics of a nonlinear system is approximated with a directed graph. Subsequently, the condition for robust control invariance is derived and an algorithm for computing the robust control invariant set is presented. The algorithm combines the iterative subdivision technique with the robust control invariance condition to produce outer approximations of the largest robust control invariant set at each iteration. Following this, we prove convergence of the algorithm to the largest RCIS as the iterations proceed to infinity. Based on the developed algorithms, an algorithm to compute inner approximations of the RCIS is also presented. A special case of input affine and disturbance affine systems is also considered. Finally, two numerical examples are presented to demonstrate the efficacy of the proposed method.
	
\end{abstract}

\noindent{\bf Keywords:} Nonlinear systems; Outer approximation; Inner approximation; Invariant sets; Graph theory.

\section{Introduction} \label{section_1}

The importance of invariant sets is evident in the numerous attention it has received in the control literature (see \cite{blanchini1999} for an extensive survey on this subject). In particular, it plays a fundamental role in control system design and analysis of constrained dynamical systems due to its implicit stability properties. For example, it plays a key role in the design of model predictive control strategies via the use of control invariant sets as a terminal region constraint \cite{mayne2001,cannon2003}. % \cite{mayne2001,camacho1995,chen1998,cannon2003}. %can be used to address the issue of stability, recursive feasibility and constraint satisfaction in model predictive control strategies ()

Given a dynamical system, a subset of the state space is said to be control invariant if the state of the system can be maintained within the set forever with the help of a control provided the initial states are chosen from within the set. In addition, if disturbances are considered, then the set is known as robust control invariant set (RCIS). Many results which address computational issues and algorithmic procedures exist for both deterministic and uncertain linear systems \cite{rungger2017,rakovic2005,kerrigan2001,kolmanovsky1998,gilbert1991}. Some results also exist for nonlinear systems \cite{fiacchini2010,alamo2009,bravo2005}. (Robust) control invariant sets are closely linked to viability theory \cite{aubin2009,maidens2013}, reachability analysis \cite{mitchell2005,lygeros2004} and null controllability \cite{homer2017,homer2018,homer2020}. It is worth mentioning that exact and efficient determination of invariant sets even for linear systems is still an open problem. For this reason, outer or inner approximations are often sought after.

Homer and Mhaskar presented an algorithm for estimating null controllable regions of nonlinear systems \cite{homer2017,homer2018} by enlarging an initial estimate of a control invariant set. However, the algorithm makes use of an invariance test function that require guessing an appropriate input sequence for a specified prediction horizon which may not be easy especially the case of systems which are not input affine. This was further extended to an approach that solves a reverse time-optimal control problem in \cite{homer2020}. In these cases however, the presence of disturbances were not considered. Fiacchini and coworkers also presented an algorithm based on difference of two convex (DC) functions to estimate convex robust control invariant sets of nonlinear systems \cite{fiacchini2010}. However, the algorithm requires contractivity and convergence to the largest RCIS was not provided. In \cite{mitchell2005}, a grid-based algorithm which solves a time-dependent Hamilton-Jacobi formulation was used to estimate reachable sets of continuous systems with uncertainties.

Graph theoretical methods have been successfully used in the analysis of nonlinear systems \cite{osipenko1983,eidenschink1997,mischaikow2002,szolnoki2003}. The idea is to approximate the trajectories of the dynamical system using directed graphs and then analyze the constructed graph using graph theoretical methods. This method has been used in identifying or computing periodic orbits, invariant sets, recurrent sets, Lyapunov exponents, etc (see \cite{osipenko2007} for more applications). However, all the studies focused on autonomous dynamical systems with the exception of \cite{szolnoki2003} who used graph theoretical method to determine control sets. Also, only outer approximations of the invariant sets were considered in these studies.

Motivated by the success of graph theory in the analysis of autonomous nonlinear systems, this paper presents a new algorithm for computing approximations of the largest robust control invariant set for general constrained time-invariant discrete-time uncertain nonlinear systems. In contrast to some existing methods, we do not assume polynomial dynamics, nor require contractivity nor require prior knowledge of the structure of the set. More importantly, the proposed algorithm yields an inner approximation of the largest robust control invariant set for a sufficiently high precision. By combining our derived robust control invariant condition with the subdivision technique, our algorithm shows higher efficiency as compared to the grid-based methods. Furthermore, we propose an approach similar to feedback linearization of nonlinear systems to further reduce the computational load.

The remainder of the paper is organised as follows: Section \ref{section_2} is concerned with preliminaries and problem formulation as well as a brief introduction to directed graph representation of autonomous dynamical system and invariant set investigation of the graph. In Section \ref{section_3}, we present an algorithm to compute the outer and inner approximations of the largest robust control invariant set. We also prove the convergence of the algorithm to the largest robust control invariant set and briefly discuss some computational issues in this section. We present two illustrative examples in Section \ref{section_4} and then conclude the paper in Section \ref{section_5}. \\

\noindent \textbf{Notation.} $\mathbb{Z}$ denotes the set of integers $\{\hdots,-2,-1,0,1,2,\hdots \}$. $\mathbb{Z}_+$ denotes the set of non-negative integers $\{ 0,1,2,\hdots \}$. $\{ z_k \}_{k \in \mathbb{Z}_+}$ denotes an ordered set of numbers according to $k \in \mathbb{Z}_+$ $\{ z_0, z_1, z_2, \hdots \}$. $\mathbb{B}$ denotes the unit ball in $\mathbb{R}^n$ with respect to the infinity norm. The operator $|\cdot|$ denotes the Euclidean norm of a vector. The Minkowski set addition of two sets $P,Q \subset \mathbb{R}^n$ is defined as $Q+P = \{p + q \in \mathbb{R}^n | q \in Q,~ p \in P\}$. A directed graph is denoted as $G=(V,E)$ with $V$ denoting the set of vertices of the graph and $E$ denoting the set of ordered pairs of vertices known as edges. A function $f:X \rightarrow X$ is said to be homeomorphic in $X$ if it is continuous with continuous inverse in $X$.

\section{Preliminaries} \label{section_2}

\subsection{System description and problem formulation}

In this work, we consider a class of discrete-time nonlinear systems that can be described by the following model:
\begin{equation}\label{system}
    x^+ = f(x, u, w)
\end{equation}
where $x^+ \in \mathbb{R}^n$ denotes the state at the next sampling time, $x \in \mathbb{R}^{n}$ is the state, $u \in \mathbb{R}^{m}$ represents the control input and $w \in \mathbb{R}^{n}$ denotes the unknown disturbance input. We consider that the state, control and disturbance are subject to the following constraints:
\begin{equation} \label{constraint}
    x \in X \subseteq \mathbb{R}^{n}, ~ u \in U \subseteq \mathbb{R}^{m} ~ \textrm{and} ~ w \in W \subseteq \mathbb{R}^{n}
\end{equation}

Throughout our discussion, we make the following assumptions:
\begin{assumption}[Compactness of constraints] \label{compactness_assumption}
    The sets $X$, $U$ and $W$ are compact metric spaces. In addition, we assume that $U$ and $W$ have the origin in their interiors.
\end{assumption}

\begin{assumption}[Smoothness of system] \label{continuity_assumption}
    The function $f:X \times U \times W \rightarrow X$ is a sufficiently smooth vector field in $X$.  In addition, we assume that for each $x \in X$, $u \in U$ and $w \in W$, $f(x,u,w)$ is uniquely defined.
\end{assumption}

We also recall the following definitions on forward invariant set and robust control invariant set:

\begin{definition}[Foward invariant set \cite{blanchini1999}] \label{postively_invariant_set}
    A set $R \subset X$ is said to be a forward or positively invariant set of the system $x^+=f(x)$ if for every $x \in R$, $f(x) \in R$. % for all $t \in \mathbb{Z}_+$.
\end{definition}

\begin{definition}[Robust control invariant set \cite{blanchini1999}]\label{robust_control_invariant_set}
    A set $R \subset X$ is said to be a robust control invariant set (RCIS) for system (\ref{system}) and constraint set (\ref{constraint}) if for every $x \in R$, there exist a feedback control law $u = \mu(x) \in U$ such that $R$ is forward invariant for the closed-loop system $f(x,u,w)$ for all $w \in W$. 
\end{definition}

In the absence of disturbances, the robust control invariant set $R$, falls back to the usual notion of control invariant set. In this work, we are interested in determining the largest robust control invariant set contained in $X$. This set is also known as the maximal RCIS, strongly reachable set \cite{bertsekas1972} or discriminating kernel \cite{cardaliaguet1994}. We denote by $R(X)$ the largest RCIS of (\ref{system}) with constraints (\ref{constraint}). 

\subsection{Graph construction}

In this section, we briefly present the notion of symbolic image of a dynamical system and its construction for autonomous systems. For a more detailed discussion, the reader may refer to Chapters 2 and 5 of \cite{osipenko2007}. 

Consider a discrete-time autonomous system as follows:
\begin{equation} \label{auto_system}
    x^+ = \hat{f}(x)
\end{equation}
where $\hat{f}:X \rightarrow X$ is a homeomorphism on a compact domain $X \subset \mathbb{R}^{n}$. 

A symbolic image of (\ref{auto_system}) is a finite approximation of the dynamics of the system using a directed graph. To construct a symbolic image, the state space $X$ is quantized with the help of a finite covering, $\mathcal{C} = \{ B_1, \hdots, B_l \}$, of the state space $X$. The finite covering $\mathcal{C}$ is a collection of closed sets known as cells or boxes $B_i,i=1,\hdots l$, such that
\begin{subequations}
    \begin{align}
         & X \subseteq \cup_{B_i \in \mathcal{C}} B_i  \\ 
         & B_i \cap B_j = \emptyset, ~ \forall B_i,B_j \in \mathcal{C} ~ \text{with} ~ i \neq j 
    \end{align}
\end{subequations}
The diameter of the covering $\mathcal{C}$ is given by
\begin{equation*}
    \text{diam}(\mathcal{C}) := \max_{B_i \in \mathcal{C}}\text{diam}(B_i)
\end{equation*}
where diam($B_i$) = $\text{sup}\{ |x-y|:x,y \in B_i \}$. Since $X$ is compact, it is always possible to obtain a finite covering. We now introduce the symbolic image of $\hat{f}$ with respect to the covering $\mathcal{C}$.

\begin{definition}[Symbolic image \cite{osipenko2007}]
    Let $G$ be a directed graph with $l$ vertices where each vertex is a cell or box $B_i$ in a finite covering $\mathcal{C}$ of the domain $X$ of system (\ref{auto_system}). The vertices $B_i$ and $B_j$ are connected by a directed edge $B_i \rightarrow B_j$ if 
    \begin{equation*}
        B_j \cap \hat{f}(B_i) \neq \emptyset
    \end{equation*}
    where $\hat{f}(B_i):=\{y | y= \hat{f}(x), x \in B_i\}$. The graph $G$ is called a symbolic image of (\ref{auto_system}) with respect to the covering $\mathcal{C}$.
\end{definition}

\begin{definition}[Admissible path \cite{osipenko2007}] \label{admissible_path}
    A sequence $\{ z_k \}_{k\in \mathbb{Z}_+}$ with each element $z_k$ taking a value from the set of vertices of $G$ is called an admissible path if for each $k \in \mathbb{Z}_+$, the graph $G$ contains the edge $z_k \rightarrow z_{k+1}$.
\end{definition}

To understand the relationship between an admissible path on the symbolic image and the trajectories of (\ref{auto_system}), we recall the notion of $\varepsilon$-orbit.

\begin{definition}[$\varepsilon$-orbit \cite{sakai1994}]
    For a given $\varepsilon > 0$, a sequence of points $\{ x_k \}_{k \in \mathbb{Z}_+}$ in $X$ is called an $\varepsilon$-orbit of system (\ref{auto_system}) if for any $k \in \mathbb{Z}_+$
    \begin{equation*}
        |\hat{f}(x_k) - x_{k+1}| < \varepsilon 
    \end{equation*}
\end{definition}

 Due to round off errors in numerically computed trajectory of system (\ref{auto_system}), its real trajectory is rarely known in practice. Thus, a numerically computed trajectory of system (\ref{auto_system}) is usually no more than an $\varepsilon$-orbit for sufficiently small positive $\varepsilon$. There is therefore a natural correspondence between admissible paths on the symbolic image and the $\varepsilon$-orbits. That is, an admissible path on the graph $G$ represents an $\varepsilon$-orbit of system (\ref{auto_system}) and vice versa. Specifically, if the sequence $\{z_k\}_{k \in \mathbb{Z}_+}$ is an admissible path on the symbolic image $G$, then there exist a sequence $\{x_k, x_k \in z_k\}_{k \in \mathbb{Z}_+}$ that is an $\varepsilon$-orbit of system (\ref{auto_system}) such that the following inequality hold
 \begin{equation*}
    |\hat{f}(x_k)-x_{k+1}| \leq \text{diam}(z_{k+1}) < \varepsilon
 \end{equation*}
It is obvious that the finer the covering, the more precise the approximation of the system trajectories. 

\begin{definition}[Out-degree of a vertex]
        The out-degree of a vertex in a directed graph is the number of edges going out of the vertex.
\end{definition}

If a vertex ($B_i$) of the symbolic image of system (\ref{auto_system}) has zero out-degree, then its image $\hat{f}(B_i)$ has no intersection with any other vertex on the symbolic image. i.e. $\hat{f}(B_i) \cap X = \emptyset$. Therefore its image $\hat{f}(B_i)$ lies outside $X$. This implies that any trajectory starting from the cell will exit the state constraint $X$ in finite time. An admissible path on the symbolic image $G$ may either be finite or infinite. An admissible path on the symbolic image is finite if it ends with a vertex that has zero out-degree. Otherwise, it is infinite. 

The construction of the symbolic image is depicted in Example \ref{symbolic_image_construction_example}.

\begin{example}[Construction of the symbolic image] \label{symbolic_image_construction_example}
    Let $X=\{ x \in \mathbb{R}^2:\| x \|_{\infty} \leq 2 \}$ and consider the autonomous two dimensional system defined by
    \begin{equation*}
        x^+ = \hat{f}(x) =
        \begin{bmatrix}
            x_1 + 1.5 \\
            x_2 + 1.5
        \end{bmatrix}
    \end{equation*}
\end{example}

To construct the symbolic image of the system in Example \ref{symbolic_image_construction_example}, the state constraint $X$ is first quantized. Quantization of $X$ is not unique. One of such quantizations is to use 16 cells of unit length to obtain the finite covering $\mathcal{C}= \{ B_1, \hdots, B_{16} \}$ of $X$. Since $X$ satisfies Assumption \ref{compactness_assumption}, this is possible. Thereafter, the image of each cell is obtained and used to construct the symbolic image. The quantized $X$ is shown in Figure \ref{symbolic_image}(a). The shaded cells ($B_6,B_7,B_{10},B_{11}$) in Figure \ref{symbolic_image}(a) are the image of cell $B_{13}$ i.e. they have an intersection with $\hat{f}(B_{13})$. The resulting graph is shown in Figure \ref{symbolic_image}(b). The existence of a directed edge between cells $B_{13}$ and $B_6$ implies there exist an admissible path between the two cells. Cells whose image have no intersection with any other cells will have no outgoing edges in the symbolic image i.e. out-degree will be zero.  The symbolic image is constructed by repeating this procedure for all other cells. The union of the resulting graphs form the symbolic image. 

What remains is how the image of a cell can be approximated. One approach to approximate the image of a cell is to finitely sample points in the cell and then find the image of the points. The image of the cell is the union of the image of each sampled point in the cell. Different sampling methods that can be used are shown in Figure \ref{box_sampling_types}. Another approach is the use of interval arithmetic \cite{moore2009}. This involves performing numerical computations on intervals rather than numbers. Interval arithmetic was used in \cite{bravo2005} to compute one-step reachable sets.

\begin{figure}[t] 
    \center{\includegraphics[width=0.7\textwidth]{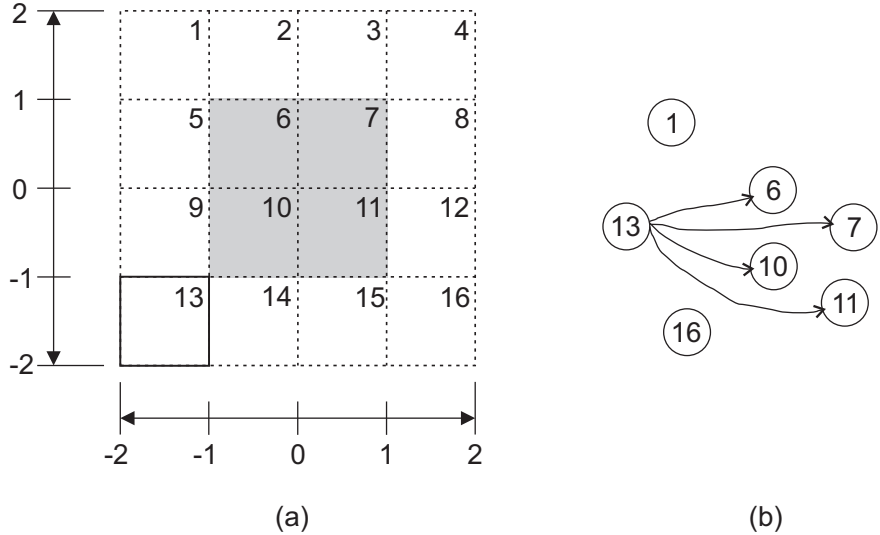}}
    \caption{\label{symbolic_image} Construction of the symbolic image. (a) Image of $B_{13}$ intersects with the shaded cells $B_6,B_7,B_{10},B_{11}$. (b) Directed graph depicting the image of $B_{13}$.}
\end{figure}

\begin{figure}[t] 
    \center{\includegraphics[width=0.7\textwidth]{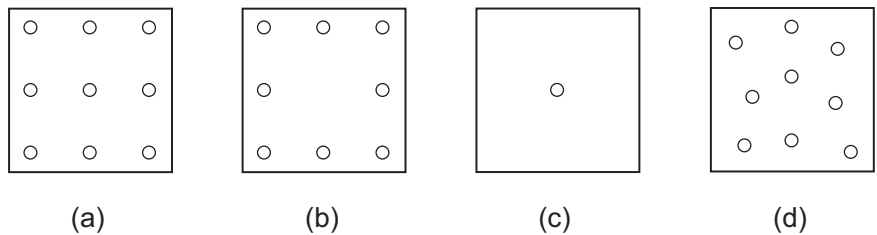}}
    \caption{\label{box_sampling_types} Different sampling types (a) Uniform sampling (b) Boundary sampling (c) Center sampling (d) Random sampling. }
\end{figure}

\subsection{Set invariance condition for autonomous systems}

In the previous section, we described the construction of the symbolic image $G$ of (\ref{auto_system}) which is an approximation of its dynamics using directed graphs. In this section, we describe how the resulting directed graph can be investigated using graph theory to obtain an outer approximation of the largest forward invariant set for autonomous dynamical systems. We recall the terminology in graph theory:

 \begin{definition}[Strongly connected graph] \label{strongly_connected}
     A directed graph $G=(V,E)$ is said to be strongly connected if there is an admissible path in both directions between each pair of vertices of the graph.
 \end{definition}
For a graph $G$ that is not strongly connected, it may contain subgraphs that are strongly connected. These subgraphs are known as the strongly connected component subgraphs of $G$. 

The following theorem summarizes how we may obtain an outer approximation of the largest forward invariant set of an autonomous system based on its symbolic image \cite{osipenko2007}.

 \begin{theorem} \label{nonleaving_autonomous}
     Let $G=(V,E)$ having a set of vertices $V$ and a set of ordered pairs of vertices $E$ be a symbolic image of the mapping $\hat{f}$ in (\ref{auto_system}) with respect to a finite convering $\mathcal{C}$ of $X$. Then 
     \begin{enumerate}[i.]
        \item \label{theorem_1_i} the vertices of the largest strongly connected component subgraph $G_s=(V_s,E_s)$ of $G$ have infinite admissible paths passing through them.
        \item \label{theorem_1_ii} any element of $V$ but not $V_s$ with a path to at least one vertex of $G_s$ also has an infinite admissible path passing through it.
        \item \label{theorem_1_iii} the union of the elements of $(i)$ and $(ii)$, $I^+(G)$, is a closed neighbourhood of the largest forward invariant set $R$ of (\ref{auto_system}) in $X$.ie.
     \begin{equation}
        R \subset I^+(G)
     \end{equation}
     \end{enumerate}
 \end{theorem}

 \begin{figure}[t] 
     \center{\includegraphics[width=0.9\textwidth]{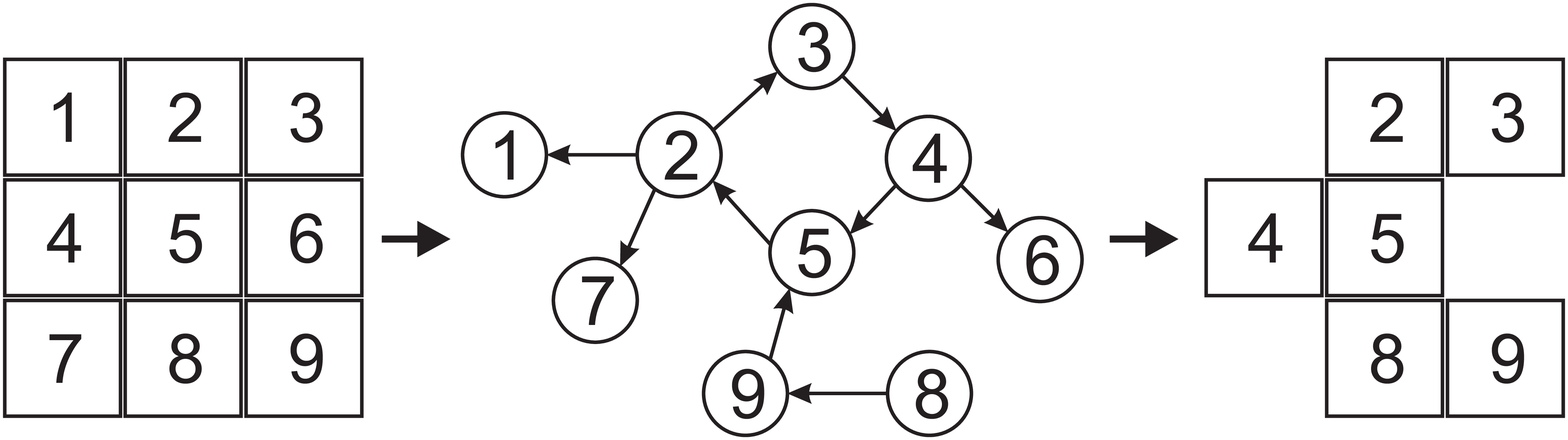}}
     \caption{\label{example_graph} Example graph construction for a ficticious autonomous system and resulting invariant set. Left: Region in state space under study; Center: Approximation of the flow of the system using directed graph. Right: Selected cells which is an outer approximation of the forward invariant set.}
 \end{figure}

Fundamentally, Theorem \ref{nonleaving_autonomous} characterizes cells that have infinite admissible paths passing through them on the symbolic image. The central idea is illustrated in Figure \ref{example_graph} for a ficticious autonomous system. In the figure, cells $B_2,B_3,B_4$ and $B_5$ form the strongly connected components subgraph of the symbolic image since there exist an admissible path in both directions between each pair of vertices on the symbolic image (Theorem \ref{nonleaving_autonomous}: \ref{theorem_1_i}). Also, since there exist an admissible path from cells $B_8$ and $B_9$ to an element (Cell $B_5$) of the the strongly connected components subgraph, they also have infinite admissible paths passing through them (Theorem \ref{nonleaving_autonomous}: \ref{theorem_1_ii}). Notice that Cell $B_2$ also has admissible path to Cells $B_1$ and $B_7$ which have zero outdegree. Similarly, there exist an admissible path from Cell $B_4$ to Cell $B_6$. Thus, both finite and infinite admissible paths pass through Cells $B_2$ and $B_4$. Therefore, the union of Cells $B_2,B_3,B_4,B_5,B_8$ and $B_9$ form a closed neighbourhood of the largest forward invariant set in the region of the state space under study (Theorem \ref{nonleaving_autonomous}: \ref{theorem_1_iii}). Since Cells $B_2$ and $B_4$ also have finite admissible paths passing through them, the set is not actually forward invariant and is merely an outer approximation of the largest forward invariant set contained in $X$. Finally, since Cells $B_1, B_6$ and $B_7$ have no outgoing edges (zero outdegree), any trajectory starting from them will exit the region of state space understudy in finite time and therefore do not form part of the approximation of forward invariant set.

We now proceed to present the main results of this work.

\section{Main results} \label{section_3}

In the previous section, we demonstrated how foward invariant sets can be outer approximated for autonomous dynamical systems based on graph theory. In this section, we extend this result for constained dynamical systems with controls and disturbances, and then use this result to develop efficient algorithms for computing robust control invariant sets. Recall that the dynamics of a system needs to be transcribed in a directed graph before analysis using graph algorithms. While the graph construction is straight forward to do for autonomous dynamical systems (\ref{auto_system}), it is not the case for dynamical systems with controls and disturbances (\ref{system}). We show how the directed graph can be constructed for system (\ref{system}) for our intended purposes as well as the robust control invariance condition based on the constructed graph.

We begin by presenting the directed graph construction and the robust control invariant set condition for constrained controlled systems with disturbances. Then, we present the algorithm for computing the robust control invariant set. Finally, we present a way to reduce the computational load for input affine systems which are a special case of system (\ref{system}).

\subsection{Robust control invariance condition}

Let us consider the set-valued map, also called parametized map,
\begin{equation} \label{differential_inclusion}
    F(x,w) := f(x,U,w) =  \{f(x,u,w) \}_{\cup_{u \in U}}
\end{equation}
The map $F$ associates with each state $x$ and disturbance $w$ the subset $F(x,w)$ of feasible next states. Therefore, system (\ref{system}) defined by the family of parametized difference equations is actually governed by the difference inclusion

\begin{equation}\label{system_difference_inclusion}
    x^+ \in F(x,w)
\end{equation}

With this difference inclusion, all feasible trajectories of (\ref{system}) under every initial state and disturbance can be obtained.  However, care must be taken when constructing the graph. If one naively constructs the symbolic image of system (\ref{system}) by considering the union of all feasible trajectories of (\ref{system_difference_inclusion}) for every initial state and disturbance, the graph obtained will not be suitable for our purposes. This is because the so-called ``best-case" will be included. The ``best-case" are instances where the disturbances aid the control inputs. They must be avoided in the graph construction as it does not provide guaranteed existence of a control law that will keep the system in the constraint set under every disturbance realization. Thus, to ensure guaranteed existence of a control law to keep the states within constraint set (\ref{constraint}) at all times, not all trajectories must be allowed on the symbolic image which will be investigated. A pessimistic or worst-case view must therefore be adopted in the construction of the symbolic image. In short, the control seeks to enlarge the robust control invariant set while the disturbance seeks to make it smaller.

To address this, we construct individual graphs for each $w \in W$ using the difference inclusion (\ref{system_difference_inclusion}) and analyze them using Theorem \ref{nonleaving_autonomous}. By taking an intersection of the resulting sets for every $w \in W$, an outer approximation of the robust control invariant set is obtained. Let us define $G_w = (V_w,E_w)$ as the symbolic image of $F(X,w)$ with respect to the finite covering $\mathcal{C}$ of $X$ and the constraint sets (\ref{constraint}). Also, let
\begin{equation} \label{robust_control_invariance_condition}
    K := \cap_{w \in W} I^+(G_w)
\end{equation} 
The set $K$ represents the cells with infinite admissible paths passing through them irrespective of the disturbance realization. The following theorem characterises an outer approximation of the robust control invariant set of (\ref{system}) with constraint set (\ref{constraint}).
\begin{theorem} \label{nonleaving_robust}
    Let $K$ be defined as in Equation \ref{robust_control_invariance_condition} above. If Assumptions \ref{compactness_assumption} and \ref{continuity_assumption} hold, then the set $K$ is a closed neighbourhood of the largest robust control invariant set $R(X)$ of system (\ref{system}) with constraint sets (\ref{constraint}) i.e.
    \begin{equation}
        R(X) \subset K
    \end{equation}
\end{theorem}

\begin{proof}
    We prove by contradition. Assume there exists a state $x \in R(X)$ but $x \notin K$. Then by definition of $K$, it implies that there exist Cells $B_i \subset X$ containing $x$ without any infinite admissible path passing through it and $B_j \subset X$ with no outgoing edge (zero outdegree) such that there exist a disturbance sequence $\{ w_0,\hdots,w_k \}$ for every input sequence $\{u_0,\hdots,u_k\}$ such that the finite sequence $\{z_0, \hdots, z_k\}$ with $z_0=B_i$ and $z_k=B_j$ exist. This implies that the trajectory of $x$ will eventually escape from $X$. This however, contradicts the assumption that $x \in R(X)$.
\end{proof}

Theorem \ref{nonleaving_robust} characterises the cells whose union forms an outer approximation of the robust control invariant set of system (\ref{system}) and constraint sets (\ref{constraint}) given a quantized state constraint set. This is the central idea we use in designing an algorithm for computing the largest robust control invariant set contained in $X$.

\subsection{Computation of robust control invariant set}

In this section, we present the algorithm for computing the robust control invariant set of system (\ref{system}) subject to constraint set (\ref{constraint}) and prove its convergence to the maximal RCIS. The algorithm is combined with the subdivision process (see \cite{dellnitz2002} for more discussion) to improve the computational efficiency. This is realized using a family of finite coverings of the RCIS starting from the state constraint. In each step, a set of cells are selected according to Theorem \ref{nonleaving_robust} and then subdivided while the remainder are discarded. The algorithm is achieved in three main steps namely subdivision, graph construction and selection. Considering the $k$-th iteration in the proposed algorithm, the operations are outlined below.
\begin{itemize}
    \item In the subdivision step, a finer covering of the RCIS is generated by dividing the current cells along one of the dimensions. If $\hat{\mathcal{C}}_{d_k}$ and $\mathcal{C}_{d_{k-1}}$ are coverings of the RCIS where $d_k$ and $d_{k-1}$ denote their respective diameters, then $d_k > d_{k-1}$ and    
    \begin{equation*}
        \cup_{B \in \hat{\mathcal{C}}_{d_k}} B = \cup_{B \in \mathcal{C}_{d_{k-1}}} B
    \end{equation*}
    The set that is subdivided does not change other than have cells with smaller diameter. In each iteration of the algorithm, the dimension along which the cells are divided is cycled. The function $subdivide()$ is used to compute the subdivision.

    \item Following the subdivision step, the graph construction step is conducted. This is achieved by creating a collection of graphs $G_k = \{ G_w^k = (V_w^k, E_w^k) ~ \forall w \in W\}$ with respect to the covering obtained from the subdivision step $\mathcal{C}_{d_k}$ where
        \begin{align*}
            V^k_{w} &= \hat{\mathcal{C}}_{d_k} ~ and \\
            E^k_{w} &= \{ (B_i, B_j) \in \hat{\mathcal{C}}_{d_k} \times \hat{\mathcal{C}}_{d_k} ~ | ~ F(B_i, w) \cap B_j \neq \emptyset  \}
        \end{align*}
        This is realized in the algorithm as the function $graph()$.
    \item Finally, the selection step involves the selection of the set of cells that have infinite paths passing through them irrespective of $w \in W$ using the robust control invariant set condition in Theorem \ref{nonleaving_robust}. i.e.
        \begin{equation*}
                        \mathcal{C}_{d_k} = \{ B \in \hat{\mathcal{C}}_{d_k} ~|~ B \in \bigcap_{G \in G_k} I^+(G) \}
        \end{equation*}
        The cells that are not selected are discarded while the selected ones goes on to the next iteration. This is represented in the algorithm as the function $select()$.
\end{itemize}

The complete algorithm is summarized in Algorithm \ref{algorithm_1} below:

\begin{algorithm} \label{algorithm_1}

\caption{Computing maximal robust control invariant set}
\KwIn{system (\ref{system}), constraint sets (\ref{constraint}) and maximum number of iterations $N$}
\KwOut{largest robust control invariant set}
$\mathcal{C}_{d_0} \leftarrow X$ \tcp{Initialization}

\For{$k \leftarrow 1,2,3,\cdots, N$}{
    \If{$\mathcal{C}_{d_{k}} = \emptyset$}{ 
        $\mathcal{C}_{d_N} \leftarrow \mathcal{C}_{d_{k}}$ \\
        \textbf{break}       
    }

    \If{$\mathcal{C}_{d_{k}} = \mathcal{C}_{d_{k-1}} $}{
        $\mathcal{C}_{d_N} \leftarrow \mathcal{C}_{d_{k}}$ \\
        \textbf{break}       
    }

    $\hat{\mathcal{C}}_{d_{k}} \leftarrow subdivide(\mathcal{C}_{d_{k}})$ \\
    $G_{k} \leftarrow graph(\hat{\mathcal{C}}_{d_k}) $ \\
    $\mathcal{C}_{d_{k+1}} \leftarrow select(G_k)$

}
\textbf{return $\mathcal{C}_{d_N}$}

\end{algorithm}

Algorithm \ref{algorithm_1} is initialized using the state constraint $X$. This not only ensures that the domain understudy is restricted to $X$, but also ensures that the state constraints are enforced. Also, since Equation \ref{system_difference_inclusion} is used during the graph construction step, the input and the disturbance constraint sets need to be finitely sampled for numerical implementation of the algorithm. An alternative is to treat the input and disturbance sets as intervals and use interval arithmetic for the computations. Notice that the computational load of the algorithm depends heavily on the number of cells generated at each iteration. This grows exponentially as the algorithm progresses. Similar to other algorithms for numerically computing invariant sets, there is a trade off between compuational load and accuracy.

In what follows, we prove the convergence of the algorithm to the largest robust control invariant set.

\subsection{Convergence of algorithm}

We now prove that Algorithm \ref{algorithm_1} always converge to the largest RCIS $R(X)$ provided $k$ goes to infinity. Let us denote by $R_k$ the collection of closed sets after the $k$-th iteration of Algorithm \ref{algorithm_1}. i.e.
\begin{equation*}
    R_k = \cup_{B \in \mathcal{C}_{d_k}} B
\end{equation*}
A quick observation is that Algorithm \ref{algorithm_1} generates a nested sequence $\{R_k\}$ of compact sets with $R_k \subset R_{k-1}$ due to the continuity of system (\ref{system}). We can therefore observe that the $N$-th output from the algorithm is given by
\begin{equation}
    R_N = \cap_{k=0}^N R_k
\end{equation}
and we may write
\begin{equation} \label{c_infinity}
    R_{\infty} = \cap_{k=0}^\infty R_k
\end{equation}
as the limit set of the algorithm. Our goal is to show that the robust control invariant set $R(X)$ is a subset of $R_{\infty}$ and vice versa. 

We first begin by showing that the sets $R_k$ contain the robust control invariant set.

\begin{lemma} \label{rcis_subset}
    Consider system (\ref{system}) with constaint sets (\ref{constraint}). If Assumptions \ref{compactness_assumption} and \ref{continuity_assumption} hold, then the sets $R_k$ obtained at the k-th iteration of Algorithm \ref{algorithm_1} contain the largest robust control invariant set. i.e.
    \begin{equation*}
        R(X) \subset R_k, ~ \forall k \in \mathbb{Z}_+
    \end{equation*}
\end{lemma}

\begin{proof}
    Obviously, we know that $R(X) \subset X = R_0$. It also follows from Theorem \ref{nonleaving_robust} that $R(X) \subset R_k \subset R_0$. Therefore $R(X) \subset R_k ~ \forall k \in \mathbb{Z}_+$.
\end{proof}

We now show that the limit set of Algorithm \ref{algorithm_1} is robust control invariant.
\begin{lemma} \label{rcis_superset}
    Consider system (\ref{system}) with constaint sets (\ref{constraint}). If Assumptions \ref{compactness_assumption} and \ref{continuity_assumption} hold, then the limit set $R_{\infty}$ obtained from Algorithm \ref{algorithm_1} is robust control invariant.
\end{lemma}

\begin{proof}
    Recall that an admissible path on a symbolic image represents an $\varepsilon$-orbit of (\ref{system}). From the construction of Algorithm \ref{algorithm_1}, we have that $d_k \rightarrow 0$ as $k \rightarrow \infty$. As a consequence of the weak shadowing property of an admissible path on the symbolic image (see \cite{osipenko2007}, Theorem 14, 2) and the continuity of system (\ref{system}), $\varepsilon \rightarrow 0$. 
	Following similar arguments in \cite{osipenko2007} (Theorem 41, 2), we have that as $\varepsilon \rightarrow 0$ every admissible path approach the true path. Now we prove by contradiction. Suppose there exists an $x \in R_{\infty}$ such that there exist $w \in W$ for every $u \in U$ such that $f(x,u,w) \notin R_{\infty}$. This implies that only finite admissible paths pass through $x$. But this is impossible by the construction of  Algorithm \ref{algorithm_1} since infinite admissible paths pass through every $x \in R_{\infty}$, and we have obtained the desired contradiction.
\end{proof}

Finally by combining Lemmas \ref{rcis_subset} and \ref{rcis_superset}, we obtain the desired convergence result for Algorithm \ref{algorithm_1}.
\begin{theorem}
    Let $R(X)$ be the largest robust control invariant set of system (\ref{system}) with constraint sets (\ref{constraint}). Consider the sequence $\{ R_k \}_{k \in \mathbb{Z}_+}$ generated by Algorithm \ref{algorithm_1}. If Assumptions \ref{compactness_assumption} and \ref{continuity_assumption} hold, then
    \begin{equation*}
        R(X) = R_{\infty}
    \end{equation*}
\end{theorem}
\begin{proof}
    To prove the above assertion, we need to show that the robust control invariant set $R(X)$ is a subset of $R_{\infty}$ and at the same time a superset of $R_{\infty}$. Now we proceed with the proof. \\
	First we show that $R(X)$ is a subset of $R_{\infty}$. It immediately follows from Lemma \ref{rcis_subset} that $R(X)$ is contained in every $R_k$ and therefore is also contained in $R_{\infty}$. i.e.
	\begin{equation*}
		R(X) \subset R_{\infty}
	\end{equation*}
	Next we show that $R_{\infty}$ is a subset of $R(X)$. It follows from Lemma \ref{rcis_superset} that the compact set $R_{\infty}$ is robust control invariant and therefore must be contained in $R(X)$. i.e.
	\begin{equation*}
		R_{\infty} \subset R(X)
	\end{equation*}

	Thus, combining the two Lemmas gives the desired result. i.e.
	\begin{equation*}
		R_{\infty} \subset R(X) \subset R_{\infty}
	\end{equation*}
    This completes the proof.
\end{proof}

\subsection{Inner approximation}
Notice that by the construction of Algorithm \ref{algorithm_1}, the sets $R_k$ are outer approximations of the largest robust control invariant set contained in $X$ (see Lemma \ref{rcis_subset}). Therefore, though the convergence results show that $R(X)$ can be computed, in practice it is impossible to infinitely go on with the construction of an arbitrary fine covering of $X$. Thus, we have that 
\begin{equation} \label{outer_inclusion}
    R_k \subseteq R(X) + \varepsilon \mathbb{B}
\end{equation}
at the $k$-th iteration of Algorithm \ref{algorithm_1}. While the set obtained gives an idea of the location and structure of the robust control invariant set of system (\ref{system}) contained $X$, in the context of control theory, the sets $R_k$ obtained from Algorithm \ref{algorithm_1} are not very useful since they are not robust control invariant. Therefore, an inner approximation is desired for controller design purposes.

One approach often used in the dynamic programming type algorithm is to initialize the algorithm from a robust control invariant set and gradually enlarge it (see \cite{blanchini2015,fiacchini2010}). An alternative approach is the use of contractive sets. The former is not applicable to our algorithm since the feasible trajections of the system (\ref{system}) is approximated within the state constraint set and the latter assumes that $X$ must contain a convex $\lambda$-contractive set and therefore restrictive. Recall that we do not assume contractivity in our analysis.

We take an approach similar to the stopping criterion of \cite{rungger2017} for obtaining the inner approximation. Thus to obtain an inner approximation using Algorithm \ref{algorithm_1}, we modify system (\ref{system}) to 
\begin{equation} \label{modified_system}
    x^+ = f(x,u,w) + \varepsilon \mathbb{B}
\end{equation}
and show that there exist $k \in \mathbb{Z}_+$ such that
\begin{equation}\label{inner_condition}
    R^\varepsilon_k  \subseteq R_{k+1}^\varepsilon + \varepsilon \mathbb{B}    
\end{equation}
where $R^\varepsilon_k$ denotes the $k$-th set generated by Algorithm \ref{algorithm_1} with the graphs constructed using the modified system (\ref{modified_system}).

The algorithm for inner approximation is presented in Algorithm \ref{algorithm_2}. The termination criterion ensures that an inner approximation is obtained.
\begin{algorithm} \label{algorithm_2}

\caption{Computing inner approximation of maximal robust control invariant set}
\KwIn{system (\ref{modified_system}), constraint sets (\ref{constraint}) and $\varepsilon$}
\KwOut{Inner approximation of largest robust control invariant set}
$\mathcal{C}_{d_0} \leftarrow X$ \tcp{Initialization}

\For{$k \leftarrow 1,2,3,\cdots$}{
    \If{$\mathcal{C}_{d_{k}} = \emptyset$}{ 
        $\mathcal{C}_{d_N} \leftarrow \mathcal{C}_{d_{k}}$ \\
        \textbf{break}       
    }

    \If{$\mathcal{C}_{d_{k-1}} \subseteq \mathcal{C}_{d_{k}} + \varepsilon \mathbb{B} $}{
        $\mathcal{C}_{d_N} \leftarrow \mathcal{C}_{d_{k}}$ \\
        \textbf{break}       
    }

    $\hat{\mathcal{C}}_{d_k} \leftarrow subdivide(\mathcal{C}_{d_{k}})$ \\
    $G_k \leftarrow graph(\hat{\mathcal{C}}_{d_k}) $ \\
    $\mathcal{C}_{d_{k+1}} \leftarrow select(G_k)$

}
\textbf{return $\mathcal{C}_{d_N}$}

\end{algorithm}

In what follows, we show that the output of Algorithm \ref{algorithm_2} is an inner approximation of the largest robust control invariant set.

\begin{lemma}
    Consider system (\ref{modified_system}) and constraint sets (\ref{constraint}). Let $\{ R_k^\varepsilon, k \in \mathbb{Z}_+\}$ be a sequence obtained from Algorithm \ref{algorithm_2}. If Assumptions \ref{compactness_assumption} and \ref{continuity_assumption} hold, then for any $\varepsilon > 0$, there exist $k \in \mathbb{Z}_+$ so that (\ref{inner_condition}) holds.
\end{lemma}

\begin{proof}
    From Lemma \ref{rcis_subset} we know that the sets $\{ R_k^\varepsilon \}$ obtained from Algorithm \ref{algorithm_1} form a closed neighbourhood of $R(X)$ and therefore there exist $k^\prime \in \mathbb{Z}_+$ such that for all $k \geq k^\prime$, we have that $R_k^\varepsilon \subseteq R(X) + \varepsilon \mathbb{B}$. It then follows that $R^\varepsilon_k  \subseteq R_{k+1}^\varepsilon + \varepsilon \mathbb{B}$. 
\end{proof}

We now show that if the stopping criterion is met in Algorithm \ref{algorithm_2}, then its output is an inner approximation of the robust control invariant set. Note that we do not consider the case where the largest robust control invariant is empty since this is trivial.
\begin{theorem}
    Consider system (\ref{modified_system}) and constraint sets (\ref{constraint}) and let $k^\prime \in \mathbb{Z}_+$ be the smallest index so that (\ref{inner_condition}) hold for some $\varepsilon > 0$. If Assumptions \ref{compactness_assumption} and \ref{continuity_assumption} hold, then for any $k \geq k^\prime$, the set $R_{k+1}^\varepsilon$ is robust control invariant.
\end{theorem}

\begin{proof}
    Let $x \in R_{k+1}^\varepsilon$ such that $k \geq k^\prime$. Then for all $w \in W$ there exist $u \in U$ such that 
        \begin{equation*}
            f(x,u,w) + \varepsilon \mathbb{B} \subseteq R_k^\varepsilon \subseteq R_{k+1}^\varepsilon + \varepsilon \mathbb{B}
        \end{equation*}
        This implies that $ f(x,u,w) \subseteq R_{k+1}^\varepsilon$ and therefore $R_{k+1}^\varepsilon$ is robust control invariant.
\end{proof}

Notice that if the disturbance is affine in $w$, then the modication in (\ref{modified_system}) is equivalent to increasing the disturbance set $W$ by $\varepsilon$ i.e. $W_{\varepsilon} \subseteq W + \varepsilon \mathbb{B}$.

\subsection{Special case: Input and disturbance affine systems} \label{computational_issues}
The algorithm as presented require the construction of several graphs of system (\ref{system}) and constraint sets (\ref{constraint}) under different disturbance realization. The computational burden for contructing the symbolic images may be excessive. Therefore in this section, we provide a simple method to reduce the computational load. The key idea is to cancel some or all of the disturbances by employing concepts from feedback linearization of nonlinear systems. 

To achieve this, we require that the discrete-time uncertain control system (\ref{system}) to have the following structure
\begin{equation}\label{input_affine}
	x^+ = f(x) + g(x)u + h(x)w
\end{equation}
where $f(\cdot) \in \mathbb{R}^{n}$, $g(\cdot) \in \mathbb{R}^{n} \times \mathbb{R}^{m}$ and $h(\cdot) \in \mathbb{R}^{n} \times \mathbb{R}^{n}$. This is possible if system (\ref{system}) is affine in the input and the disturbance. 
If the state equation takes the form (\ref{input_affine}), then we can cancel out some or all disturbances via the equation
\begin{equation} \label{input_transformation}
	u = - g(x)^{-1} h(x) w + v
\end{equation}
to obtain the transformed equation
\begin{equation}\label{input_affine_transformed}
	x^+ = f(x) + g(x)v
\end{equation}

The caveat to using this approach is that the state dependent matrix $g(\cdot)$ must be nonsingular for every $x \in X$. If $g(\cdot)$ is nonsingular, then the bounds on $v$ can be obtained from Equation \ref{input_transformation}. Notice that if $g(\cdot)$ is independent of the state i.e. constant, then the bounds on $v$ remain constant and therefore can be determined offline prior to the start of the algorithm. However, if $g(\cdot)$ is state dependent, then the bounds on $v$ may vary for every $x \in X$ and hence has to be determined online. Unfortunately matrix inversion may not be cheap.

In cases where the $g(\cdot)$ matrix is non-square, it is possible to make it square by introducing arbitrary inputs. 
% This will be is demonstrated in the example section of this paper.

\begin{proposition} \label{equivalent_systems}
    The following constrained systems are equivalent:
    \begin{enumerate}[(i)]
    \item $x^+ = f_1(x) + g_1(x)u + h_1(x)w$ with $x \in X$, $u \in U$ and $w \in W$
    \item $x^+ = f_1(x) + g_1(x)u + g_2(x)u_a + h_1(x)w$ with $x \in X$, $u \in U$, $u_a \in \{0\}$ and $w \in W$
    \end{enumerate}
\end{proposition}

\begin{proof}
    It is obvious that if $u_a = 0$ for all $t \in \mathbb{Z}_+$, then $g_2(\cdot)$ vanishes and has no effect on the dynamics of $(ii)$. Hence, the two systems are equivalent. 
\end{proof}

% \begin{remark}
%     The 
% \end{remark}
Proposition \ref{equivalent_systems} shows that arbitrary inputs can be added to a difference equation without affecting the dynamics so long as it is constrained to the origin. This makes it possible to alter the structure of $g(\cdot)$ to a square matrix if it is not already a square. This is demonstrated using a linear system in Section \ref{section_4}.

\section{Illustrative examples} \label{section_4}
In this section we conclude our study by demonstrating the efficacy of our algorithm on two problems. Both algorithms were implemented in the numerical computation language Julia \cite{bezanson2017}. In the first example, we use a simple two-dimensional linear model to show how feedback linearization-like concepts can be used to reduce the computational load and compare the inner and outer approximation of the invariant sets generated by our algorithm with the algorithm presented in \cite{rungger2017}. In the second example, we apply the algorithm to a two-dimensional nonlinear model used in the work of \cite{fiacchini2010} and compare the estimated robust control invariant set to other works.

In both examples, 10 points close to the edges of each cell were uniformly sampled and 5 points were also uniformly sampled in $U$. Also the vertices of the disturbance set were selected since both systems are affine in the disturbance. The maximum number of iterations $N$ in Algorithm \ref{algorithm_1} and $\varepsilon$ in Algorithm \ref{algorithm_2} are set at $16$ and $0.001$ respectively.

\subsection{Example 1}
Consider the linear system
\begin{equation*}
	x^+ = Ax + Bu + Gw
\end{equation*}
where $A$, $B$, and $G$ are given by 
\begin{equation*}
	A = \begin{bmatrix} 0.0 & 1.0 \\ 1.0 & 1.0 \end{bmatrix}, B =   \begin{bmatrix} 0.0 \\ 1.0 \end{bmatrix} ~\text{and} ~ G= \begin{bmatrix} 1.0 & 0.0 \\ 0.0 & 1.0 \end{bmatrix}.
\end{equation*}
The constraints on the states and input are $X = \{ x \in \mathbb{R}^2 : \| x \|_{\infty} \leq 5 \} $ and $U = \{ u \in \mathbb{R} : \| u \|_{\infty} \leq 2 \}$ respectively. The disturbance on the other hand is restricted to the set $W = \{ w \in \mathbb{R}^2 : \| w \|_{\infty} \leq 0.3  \} $. 

As can be observed, $B$ is a vector and therefore not invertible. However from Proposition \ref{equivalent_systems}, its structure can be altered without changing the dynamics by introducing arbitrary inputs. In this case the modified equation becomes
\begin{equation}
    x^+ = Ax + \begin{bmatrix} 0.0 & 1.0 \\ 1.0 & 0.0 \end{bmatrix} \begin{bmatrix} u_1 \\ u_a \end{bmatrix} + Gw
\end{equation}
With this structural change, the system can be transformed to that of Equation \ref{input_affine_transformed} such that
\begin{equation}
    v = 
    \begin{bmatrix}
        u_1 \\
        u_a    
    \end{bmatrix}
    + 
    \begin{bmatrix}
        0.0 & 1.0 \\
        1.0 & 0.0
    \end{bmatrix}^{-1}
    \begin{bmatrix}
        1.0 & 0.0 \\
        0.0 & 1.0
    \end{bmatrix} w
\end{equation}
Since $g(\cdot) = B$ is invertible for all $x(t) \in X$ and independent of the current state, the bounds on $v(t)$ can be obtained.
\begin{equation}
    v = 
    \begin{bmatrix}
        u_1 + w_2\\
        u_a + w_1    
    \end{bmatrix}
\end{equation}
From the above equation, $|v_1| \leq 1.7$ and $|v_2| \leq 0.3$. Notice that $v_2$ is still a disturbance while $v_1$ is an input. With the transformed system having 1 disturbance, the graph construction reduces to 2 instead of 4. This cuts down the computation time by half.

We compare the results outer and inner approximations of the largest robust control invariant set obtained by Algorithms \ref{algorithm_1} and \ref{algorithm_2} respectively, to that of \cite{rungger2017}. We allowed the latter to run for a suffiently long time to ensure its as close as possible to the actual invariant set. As can be observed in Figure \ref{robust_control_invariant_sets_linear}, our proposed algorithms are able to provide inner and outer approximations of the largest robust control invariant.

\begin{figure}[t] 
    \center{\includegraphics[width=\textwidth]{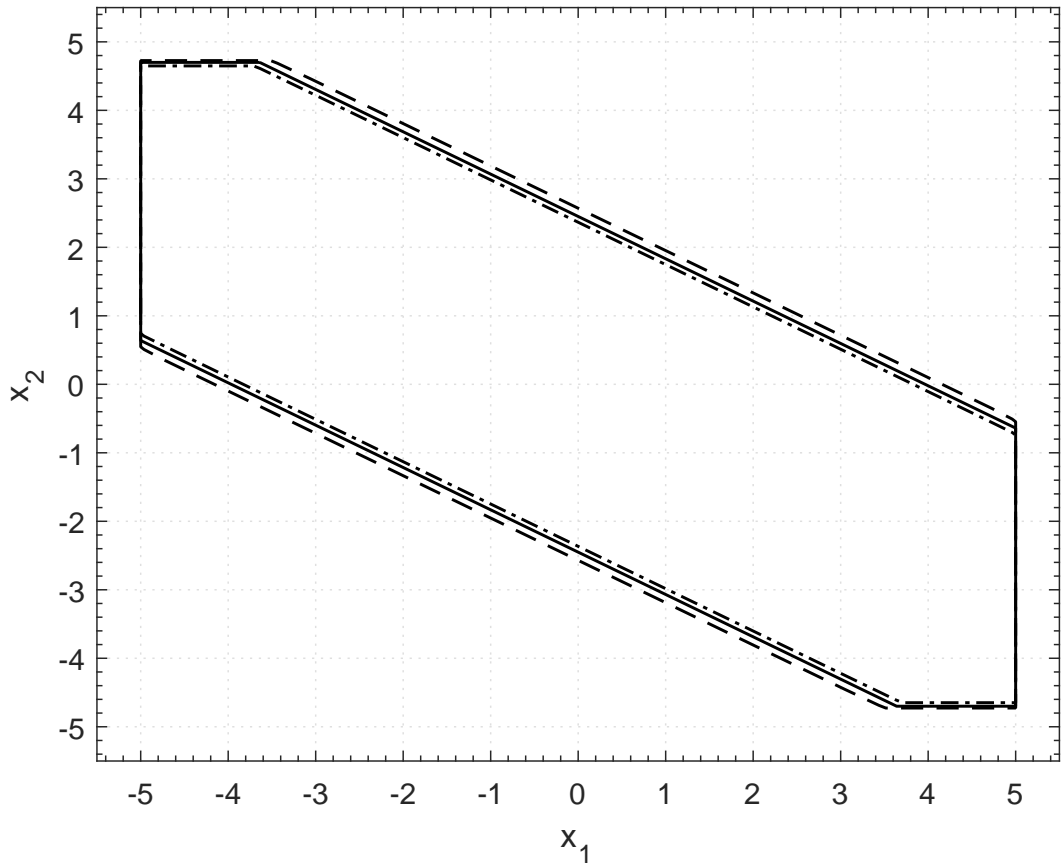}}
    \caption{\label{robust_control_invariant_sets_linear} Comparison of inner (dashed dot) and outer approximations (dashes) obtained from Algorithm \ref{algorithm_1} and that of \cite{rungger2017} (solid) for a two dimensional linear system. The invarinat sets in the figure are obtained by finding the convexhull of the cells obtained from the Algorithms.}
\end{figure}

\subsection{Example 2}
Consider the nonlinear system
{\footnotesize
\begin{equation*}\label{eqn:linear_model}
    \begin{bmatrix} 
        x_1^+ \\
        x_2^+
    \end{bmatrix}
    = 
    \begin{bmatrix} 
        1.0 & T \\
        T & 1.0
    \end{bmatrix}
    \begin{bmatrix} 
        x_1 \\
        x_2
    \end{bmatrix}
    + 
    T
    \left\{
    \mu
    \begin{bmatrix} 
        1.0 \\
        1.0
    \end{bmatrix}
    + (1-\mu)
    \begin{bmatrix} 
        1.0 & 0.0\\
        0.0 & -4.0
    \end{bmatrix} x
    \right\}
    u
    + 
    T
    \begin{bmatrix} 
        1.0 & 0.0\\
        0.0 & 1.0 
    \end{bmatrix}
    \begin{bmatrix} 
        w_1 \\
        w_2
    \end{bmatrix}
\end{equation*}
where $T=0.01$, $\mu = 0.9$.
}
The constraints on the states and input are $X = \{ x \in \mathbb{R}^2 : \| x \|_{\infty} \leq 4  \} $ and $U = \{ u \in \mathbb{R} : \| u \|_{\infty} \leq 2  \}$ respectively. The disturbance on the other hand is restricted to the set $W = \{ w \in \mathbb{R}^2 : \| w \|_{\infty} \leq 0.4  \} $.

Figure \ref{nonlinear_comparison} shows the comparison between the inner approximation of the control invariant set obtained after the 16th subvidision from Algorithm \ref{algorithm_2} and that of the \cite{fiacchini2010}. It also shows the robust control invariant set obtained from Algorithm \ref{algorithm_2}.
\begin{figure}[t] 
    \center{\includegraphics[width=\textwidth]{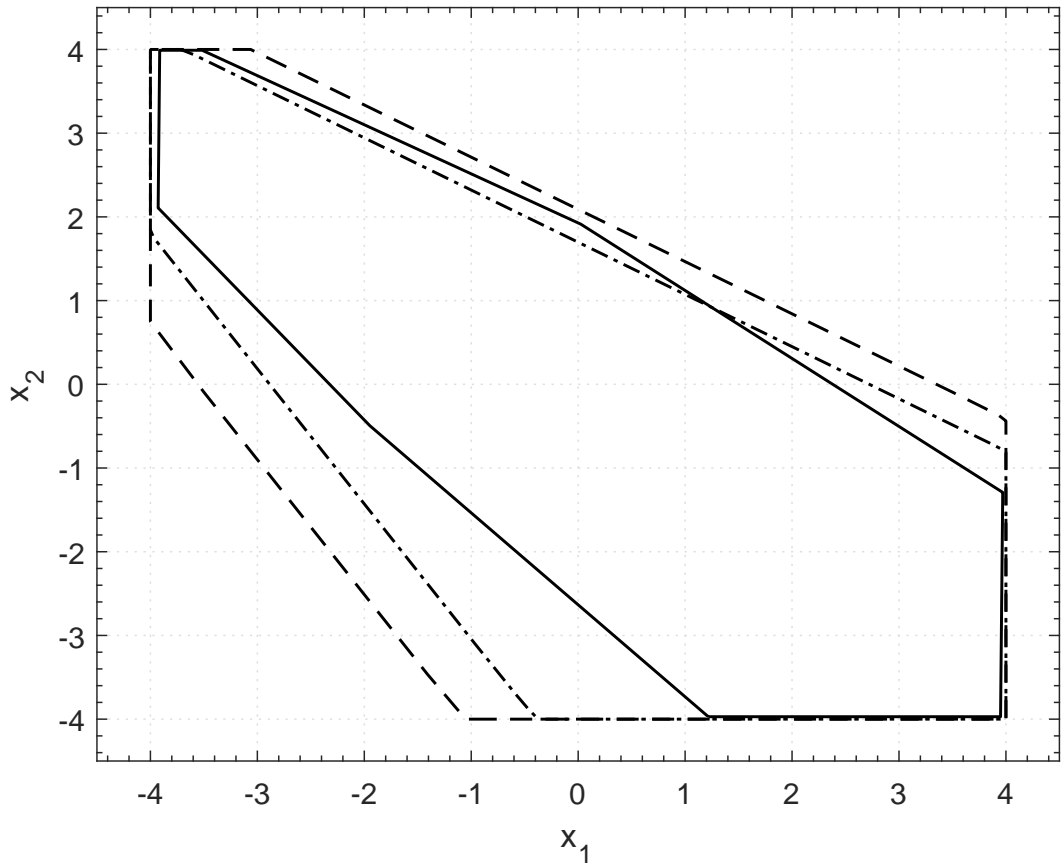}}
    \caption{ \label{nonlinear_comparison} Comparison of inner approximations of control invariant set (dashes), robust control invariant set (dashed dot) obtained from Algorithm \ref{algorithm_2} and control invariant set obtained from \cite{fiacchini2010} (solid) for a two dimensional nonlinear system. The invarinat sets in the figure are obtained by finding the convexhull of the cells obtained in Algorithm \ref{algorithm_2}.}
\end{figure}
As can be seen in Figure \ref{nonlinear_comparison}, our Algorithm was able to compute a much larger control invariant set compared to that of \cite{fiacchini2010}. 

\section{Concluding remarks} \label{section_5}
Given a discrete-time time-invariant uncertain control system with bounded disturbances subject to state and input constraints, this paper obtains inner and outer approximations of the largest robust control invariant set contained in the state constraint. For this purpose, we presented an algorithm which approximates the dynamics of the uncertain control system as a directed graph allowing for analysis of the system using graph theory. The results of this work, proving convergence to the largest robust control invariant set, are important in that the show the theoretical soundness of this approach. Simulations using linear and nonlinear systems demonstrate the effectiveness of the proposed method.

Future work will focus on the computational efficiency of the proposed algorithm. Since robust control invariant sets typically have full dimension in the state constraint, it will be worth exploring ways to identify cells that does not need to be subdivided thus keeping the number of cells generated low as well as extend this approach to compute invariant sets of higher dimensional uncertain systems.

% \section*{References}
\bibliographystyle{ieeetr}
\bibliography{ref}

\begin{thebibliography}{10}

\bibitem{blanchini1999}
F.~Blanchini, ``Set invariance in control,'' {\em Automatica}, vol.~35, no.~11,
  pp.~1747--1767, 1999.

\bibitem{mayne2001}
D.~Q. Mayne, ``Control of constrained dynamic systems,'' {\em European Journal
  of Control}, vol.~7, no.~2-3, pp.~87--99, 2001.

\bibitem{cannon2003}
M.~Cannon, V.~Deshmukh, and B.~Kouvaritakis, ``Nonlinear model predictive
  control with polytopic invariant sets,'' {\em Automatica}, vol.~39, no.~8,
  pp.~1487--1494, 2003.

\bibitem{rungger2017}
M.~{Rungger} and P.~{Tabuada}, ``Computing robust controlled invariant sets of
  linear systems,'' {\em IEEE Transactions on Automatic Control}, vol.~62,
  pp.~3665--3670, July 2017.

\bibitem{rakovic2005}
S.~V. Rakovic, E.~C. Kerrigan, K.~I. Kouramas, and D.~Q. Mayne, ``Invariant
  approximations of the minimal robust positively invariant set,'' {\em IEEE
  Transactions on Automatic Control}, vol.~50, no.~3, pp.~406--410, 2005.

\bibitem{kerrigan2001}
E.~C. Kerrigan, {\em Robust constraint satisfaction: Invariant sets and
  predictive control}.
\newblock PhD thesis, University of Cambridge, 2001.

\bibitem{kolmanovsky1998}
I.~Kolmanovsky and E.~G. Gilbert, ``Theory and computation of disturbance
  invariant sets for discrete-time linear systems,'' {\em Mathematical problems
  in engineering}, vol.~4, no.~4, pp.~317--367, 1998.

\bibitem{gilbert1991}
E.~G. {Gilbert} and K.~T. {Tan}, ``Linear systems with state and control
  constraints: the theory and application of maximal output admissible sets,''
  {\em IEEE Transactions on Automatic Control}, vol.~36, pp.~1008--1020, Sep.
  1991.

\bibitem{fiacchini2010}
M.~Fiacchini, T.~Alamo, and E.~Camacho, ``On the computation of convex robust
  control invariant sets for nonlinear systems,'' {\em Automatica}, vol.~46,
  pp.~1334--1338, Aug. 2010.

\bibitem{alamo2009}
T.~Alamo, A.~Cepeda, M.~Fiacchini, and E.~F. Camacho, ``Convex invariant sets
  for discrete-time lur’e systems,'' {\em Automatica}, vol.~45, no.~4,
  pp.~1066--1071, 2009.

\bibitem{bravo2005}
J.~M. Bravo, D.~Lim{\'o}n, T.~Alamo, and E.~F. Camacho, ``On the computation of
  invariant sets for constrained nonlinear systems: An interval arithmetic
  approach,'' {\em Automatica}, vol.~41, no.~9, pp.~1583--1589, 2005.

\bibitem{aubin2009}
J.~P. Aubin, {\em Viability theory}.
\newblock Modern {Birkhäuser} classics, Boston: Birkhäuser, 2009.

\bibitem{maidens2013}
J.~N. Maidens, S.~Kaynama, I.~M. Mitchell, M.~M. Oishi, and G.~A. Dumont,
  ``Lagrangian methods for approximating the viability kernel in
  high-dimensional systems,'' {\em Automatica}, vol.~49, no.~7, pp.~2017--2029,
  2013.

\bibitem{mitchell2005}
I.~Mitchell, A.~Bayen, and C.~Tomlin, ``A time-dependent hamilton-jacobi
  formulation of reachable sets for continuous dynamic games,'' {\em {IEEE}
  Transactions on Automatic Control}, vol.~50, pp.~947--957, jul 2005.

\bibitem{lygeros2004}
J.~Lygeros, ``On reachability and minimum cost optimal control,'' {\em
  Automatica}, vol.~40, no.~6, pp.~917--927, 2004.

\bibitem{homer2017}
T.~Homer and P.~Mhaskar, ``Constrained control lyapunov function-based control
  of nonlinear systems,'' {\em Systems \& Control Letters}, vol.~110,
  pp.~55--61, 2017.

\bibitem{homer2018}
T.~Homer and P.~Mhaskar, ``Utilizing null controllable regions to stabilize
  input-constrained nonlinear systems,'' {\em Computers \& Chemical
  Engineering}, vol.~108, pp.~24--30, 2018.

\bibitem{homer2020}
T.~Homer, M.~Mahmood, and P.~Mhaskar, ``A trajectory-based method for
  constructing null controllable regions,'' {\em International Journal of
  Robust and Nonlinear Control}, vol.~30, no.~2, pp.~776--786, 2020.

\bibitem{osipenko1983}
G.~S. {Osipenko}, ``On a symbolic image of dynamical system,'' {\em Boundary
  Value Problems}, vol.~$~$, pp.~Interuniv. Collect, Sci. Works, perm (in
  Russian) 101--105, 1983.

\bibitem{eidenschink1997}
M.~Eidenschink, ``Exploring global dynamics: A numerical algorithm based on the
  conley index theory.,'' 1997.

\bibitem{mischaikow2002}
K.~Mischaikow, ``Topological techniques for efficient rigorous computation in
  dynamics,'' {\em Acta Numerica}, vol.~11, pp.~435--477, 2002.

\bibitem{szolnoki2003}
D.~Szolnoki, ``Set oriented methods for computing reachable sets and control
  sets,'' {\em Discrete and Continuous Dynamical Systems - Series B}, vol.~3,
  pp.~361--382, May 2003.

\bibitem{osipenko2007}
G.~Osipenko, {\em Dynamical Systems, Graphs, and Algorithms}.
\newblock No.~1889 in Lecture Notes in Mathematics, Berlin ; New York:
  {Springer}, 2007.
\newblock OCLC: ocm75927357.

\bibitem{bertsekas1972}
D.~{Bertsekas}, ``Infinite time reachability of state-space regions by using
  feedback control,'' {\em IEEE Transactions on Automatic Control}, vol.~17,
  pp.~604--613, October 1972.

\bibitem{cardaliaguet1994}
P.~Cardaliaguet, M.~Quincampoix, and P.~Saint-Pierre, ``Some algorithms for
  differential games with two players and one target,'' {\em ESAIM:
  Mathematical Modelling and Numerical Analysis-Mod{\'e}lisation
  Math{\'e}matique et Analyse Num{\'e}rique}, vol.~28, no.~4, pp.~441--461,
  1994.

\bibitem{sakai1994}
K.~Sakai, ``Pseudo-orbit tracing property and strong transversality of
  diffeomorphisms on closed manifolds,'' {\em Osaka Journal of Mathematics},
  vol.~31, no.~2, pp.~373--386, 1994.

\bibitem{moore2009}
R.~E. Moore, R.~B. Kearfott, and M.~J. Cloud, {\em Introduction to interval
  analysis}, vol.~110.
\newblock Siam, 2009.

\bibitem{dellnitz2002}
M.~Dellnitz and O.~Junge, ``Set {Oriented} {Numerical} {Methods} for
  {Dynamical} {Systems},'' in {\em Handbook of {Dynamical} {Systems}}, vol.~2,
  pp.~221--264, Elsevier, 2002.

\bibitem{blanchini2015}
F.~Blanchini and S.~Miani, {\em Set-Theoretic Methods in Control}.
\newblock Springer International Publishing, 2015.

\bibitem{bezanson2017}
J.~Bezanson, A.~Edelman, S.~Karpinski, and V.~B. Shah, ``Julia: A fresh
  approach to numerical computing,'' {\em SIAM review}, vol.~59, no.~1,
  pp.~65--98, 2017.

\end{thebibliography}

\end{document}